\documentclass[letterpaper,12pt]{article}
\usepackage{amssymb, amsmath}
\usepackage{amsthm}
\usepackage{algorithm}
\usepackage{algorithmic}
\usepackage{enumerate}
\usepackage[numbers,square,comma,sort&compress]{natbib}
\usepackage{a4}

\newtheorem{theorem}{Theorem}
\newtheorem{corollary}[theorem]{Corollary}
\newtheorem{fact}[theorem]{Fact}
\newtheorem{lemma}[theorem]{Lemma}

\newcommand{\comment}[1]{}
 
\newcommand\bs[1]{\boldsymbol{#1}}

\begin{document}
\title{A lower bound for metric $1$-median selection \footnote{A
preliminary version of this paper appears in {\em Proceedings of the 30th
Workshop on Combinatorial Mathematics and Computation Theory}, Hualien,
Taiwan, April 2013, pp.~65--68.}}

\author{
Ching-Lueh Chang \footnote{Department of Computer Science and
Engineering, Yuan Ze University, Taoyuan, Taiwan. Email:
clchang@saturn.yzu.edu.tw}
\footnote{Supported in part by the National Science Council
of Taiwan under
grant
101-2221-E-155-015-MY2.}
}


\maketitle

\begin{abstract}
Consider the problem of finding
a point in
an $n$-point
metric space
with the minimum average distance to
all
points.
We show
that
this problem
has no deterministic
$o(n^2)$-query $(4-\Omega(1))$-approximation algorithms.
\end{abstract}

\section{Introduction}

Given oracle access to
a
metric
space $(\{1,2,\ldots,n\},d)$,
the {\sc metric $1$-median}
problem asks for a point with the minimum average distance to all points.
Indyk~\cite{Ind99, Ind00} shows that {\sc metric $1$-median} has a
Monte-Carlo $O(n/\epsilon^2)$-time $(1+\epsilon)$-approximation algorithm
with an $\Omega(1)$ probability of success.
The more general {\sc metric $k$-median} problem asks for
$x_1$,
$x_2$, $\ldots$, $x_k\in\{1,2,\ldots,n\}$ minimizing
$\sum_{x\in\{1,2,\ldots,n\}}\,\min_{i=1}^k\,d(x_i,x)$.
Randomized
as well as
evasive
algorithms are well-studied for
{\sc metric $k$-median} and the related $k$-means problem~\cite{GMMMO03,
MP04, AGKMMP04, Che09, KSS10, JKS12},
where $k\ge 1$ is part of the input rather than a constant.

This paper focuses on {\em deterministic sublinear-query} algorithms for
{\sc metric $1$-median}.
Guha et al.~\cite[Sec.~3.1--3.2]{GMMMO03} prove
that {\sc metric $k$-median} has
a deterministic
$O(n^{1+\epsilon})$-time
$O(n^\epsilon)$-space
$2^{O(1/\epsilon)}$-approximation algorithm
that reads distances in a single pass, where $\epsilon>0$.
Chang~\cite{Cha13} presents a deterministic
nonadaptive
$O(n^{1.5})$-time $4$-approximation algorithm for {\sc metric $1$-median}.
Wu~\cite{Wu14} generalizes Chang's result by showing an
$O(n^{1+1/h})$-time
$2h$-approximation
algorithm for
any integer $h\ge 2$.
On the negative side,
Chang~\cite{Cha12}
shows that {\sc metric $1$-median}
has no
deterministic
$o(n^2)$-query
$(3-\epsilon)$-approximation algorithms for any constant
$\epsilon>0$~\cite{Cha12}.
This
paper
improves upon his result by showing
that {\sc metric $1$-median} has no deterministic
$o(n^2)$-query $(4-\epsilon)$-approximation algorithms for any
constant $\epsilon>0$.

In social network analysis,
the importance of an actor in a network
may
be
quantified by
several
centrality measures, among which the closeness centrality of an actor is
defined to be
its average distance to other actors~\cite{WF94}.
So
{\sc metric $1$-median}
can
be interpreted as the problem of finding the
most important point in a metric space.
Goldreich and Ron~\cite{GR08} and Eppstein and Wang~\cite{EW04} present
randomized algorithms for approximating the closeness centralities of
vertices in undirected graphs.

\section{Definitions}\label{definitionssection}

For $n\in\mathbb{N}$,
denote
$[n]\equiv \{1,2,\ldots,n\}$.
Trivially, $[0]=\emptyset$.
An $n$-point metric space $([n],d)$ is the set $[n]$, called the groundset,
endowed with a function
$d\colon [n]\times[n]\to\mathbb{R}$ satisfying
\begin{enumerate}[(1)]
\item\label{nonnegative}
$d(x,y)\ge 0$ (non-negativeness),
\item
$d(x,y)=0$ if and only if $x=y$ (identity of indiscernibles),
\item\label{symmetry}
$d(x,y)=d(y,x)$ (symmetry), and
\item
$d(x,y)+d(x,z)\ge d(y,z)$ (triangle inequality)
\end{enumerate}
for all $x$, $y$, $z\in [n]$.
An
equivalent definition
requires the triangle inequality only for distinct
$x$, $y$, $z\in [n]$, axioms~(\ref{nonnegative})--(\ref{symmetry}) remaining.

An algorithm with oracle access to a metric space $([n],d)$
is given $n$ and may query
$d$ on
any
$(x,y)\in[n]\times[n]$ to obtain $d(x,y)$.
Without loss of generality,
we
forbid
queries for $d(x,x)$, which trivially return $0$,
as well as repeated queries, where
a query
for $d(x,y)$
is considered to repeat that for
$d(y,x)$.
For convenience,
denote
an algorithm ALG with oracle access to $([n],d)$
by $\text{ALG}^d$.

Given oracle access to a finite metric space $([n],d)$, the
{\sc metric $1$-median} problem asks for a point
in
$[n]$
with the minimum average distance to
all
points.
An algorithm for
this problem is
$\alpha$-approximate if
it outputs
a point $x\in[n]$
satisfying
$$\sum_{y\in[n]}\,d\left(x,y\right)
\le\alpha\,\min_{x^\prime\in[n]}\,\sum_{y\in[n]}\,d\left(x^\prime,y\right),$$
where
$\alpha\ge 1$.

The following theorem is due to Chang~\cite{Cha13} and generalized by
Wu~\cite{Wu14}.

\begin{theorem}[{\cite{Cha13, Wu14}}]\label{nonadaptiveupperbound}
{\sc Metric $1$-median} has a deterministic nonadaptive $O(n^{1.5})$-time
$4$-approximation algorithm.
\end{theorem}

\section{Lower bound}

Fix arbitrarily a
deterministic
$o(n^2)$-query algorithm
$A$
for {\sc metric $1$-median}
and a constant
$\delta\in(0,0.1)$.
By padding queries, we may assume the existence of a function
$q\colon\mathbb{Z}^+\to\mathbb{Z}^+$ such that $A$ makes exactly
$q(n)=o(n^2)$ queries given oracle access to any metric space with groundset
$[n]$.

We introduce some notations concerning a function
$d\colon[n]\times[n]\to\mathbb{R}$
to be determined later.
For $i\in[q(n)]$, denote the $i$th query of $A^d$ by
$(x_i,y_i)\in[n]\times[n]$; in other
words, the $i$th query of $A^d$ asks for $d(x_i,y_i)$.
Note that $(x_i,y_i)$ depends only on $d(x_1,y_1)$, $d(x_2,y_2)$,
$\ldots$, $d(x_{i-1},y_{i-1})$ because $A$
is deterministic and
has been fixed.
For $x\in[n]$ and $i\in\{0,1,\ldots,q(n)\}$,
\begin{eqnarray}
N_i(x)
&\stackrel{\text{def.}}{=}&
\left\{
y\in[n]\mid
\left\{
\left(x,y\right), \left(y,x\right)
\right\}
\cap \left\{\left(x_j,y_j\right)\mid
j\in\left[i\right]\right\}
\neq\emptyset
\right\},
\label{neighborhoodinsubgraph}\\
\alpha_i(x)
&\stackrel{\text{def.}}{=}&
\left|\,
N_i(x)
\,\right|,
\label{numberoffrozenincidentdistances}
\end{eqnarray}
following Chang~\cite{Cha12} with a slight change in notation.
Equivalently, $\alpha_i(x)$ is the degree of $x$ in the undirected graph
with vertex set $[n]$ and edge set $\{(x_j,y_j)\mid j\in[i]\}$.
As $[0]=\emptyset$,
$\alpha_0(x)=0$ for $x\in[n]$.
Note that
$\alpha_i(\cdot)$ depends only on
$(x_1,y_1)$, $(x_2,y_2)$, $\ldots$, $(x_i,y_i)$.
Denote the
output of $A^d$
by $p$.
By adding at most $n-1=o(n^2)$ dummy queries,
we may
assume without loss of generality that
\begin{eqnarray}
\left(p,y\right)\in\left\{\left(x_i,y_i\right)\mid i\in\left[q(n)\right]\right\}
\label{algorithmoutputheavilyqueried}
\end{eqnarray}
for all $y\in[n]\setminus\{p\}$.
Consequently,
\begin{eqnarray}
\alpha_{q(n)}(p)=n-1.\label{outputallasked}
\end{eqnarray}
Fix any set $S\subseteq [n]$ of size $\lceil\delta n\rceil$,
e.g., $S=[\lceil\delta n\rceil]$.

We proceed to
construct $d$ by gradually freezing distances.
For brevity, freezing
the value of
$d(x,y)$
implicitly freezes $d(y,x)$ to the same value, where $x$, $y\in[n]$.
Inductively, having
answered the first $i-1$ queries of $A^d$ by freezing
$d(x_1,y_1)$, $d(x_2,y_2)$, $\ldots$,
$d(x_{i-1},y_{i-1})$, where $i\in[q(n)]$,
answer the $i$th query by
\begin{eqnarray}
d\left(x_i,y_i\right)
&=&\left\{
\begin{array}{ll}
3, &\text{if $x_i$, $y_i\in S$;}\\
3, &\text{if $x_i\in S$, $y_i\notin S$ and $\alpha_{i-1}(x_i)\le\delta n$;}\\
3, &\text{if $y_i\in S$, $x_i\notin S$ and $\alpha_{i-1}(y_i)\le\delta n$;}\\
4, &\text{if $x_i\in S$, $y_i\notin S$ and $\alpha_{i-1}(x_i)>\delta n$;}\\
4, &\text{if $y_i\in S$, $x_i\notin S$ and $\alpha_{i-1}(y_i)>\delta n$;}\\
2, &\text{if $x_i$, $y_i\notin S$ and $\max\{\alpha_{i-1}(x_i),\alpha_{i-1}(y_i)\}\le\delta n$;}\\
4, &\text{if $x_i$, $y_i\notin S$ and $\max\{\alpha_{i-1}(x_i),\alpha_{i-1}(y_i)\}>\delta n$.}
\end{array}
\right.
\label{distanceassignment}
\end{eqnarray}
It is not hard to verify that
the seven cases
in equation~(\ref{distanceassignment})
are
exhaustive and
mutually
exclusive.
We have now
frozen
$d(x_i,y_i)$ for all $i\in[q(n)]$ and none of the other distances.
As repeated queries are forbidden, equation~(\ref{distanceassignment})
does not freeze one distance twice,
preventing inconsistency.

Set
\begin{eqnarray}
B
&\stackrel{\text{def.}}{=}&
\left\{
x\in[n]\mid \alpha_{q(n)}(x)>\delta n
\right\},\label{badpoints}\\
\hat{p}
&\stackrel{\text{def.}}{=}&
\mathop{\rm argmin}_{x\in S}\, \alpha_{q(n)}(x),
\label{trueoptimal}
\end{eqnarray}
breaking ties arbitrarily.
For all distinct
$x$, $y\in[n]$
with $(x,y)$, $(y,x)\notin\{(x_i,y_i)\mid i\in[q(n)]\}$,
let
\begin{eqnarray}
d\left(x,y\right)
=
\left\{
\begin{array}{ll}
1, &\text{if $x=\hat{p}$, $y\notin S\cup B$;}\\
1, &\text{if $y=\hat{p}$, $x\notin S\cup B$;}\\
3, &\text{if $x$, $y\in S\cup B$;}\\
4, &\text{if $x\in (S\cup B)\setminus \{\hat{p}\}$ and $y\notin (S\cup B\cup\{\hat{p}\})$;}\\
4, &\text{if $y\in (S\cup B)\setminus \{\hat{p}\}$ and $x\notin (S\cup B\cup\{\hat{p}\})$;}\\
2, &\text{otherwise.}
\end{array}
\right.
\label{completingthemetric}
\end{eqnarray}
Clearly, the six cases in equation~(\ref{completingthemetric}) are
exhaustive and
mutually
exclusive.
Furthermore, equation~(\ref{completingthemetric}) assigns the same value
to $d(x,y)$ and $d(y,x)$.
Finally, for all $x\in[n]$,
\begin{eqnarray}
d\left(x,x\right)=0.\label{trivialdistance}
\end{eqnarray}
Equations~(\ref{distanceassignment}),~(\ref{completingthemetric})~and~(\ref{trivialdistance})
complete the construction of $d$ by freezing all distances.

The following lemma is straightforward.

\begin{lemma}\label{distancesarezeroto4}
For all distinct $x$, $y\in[n]$,
$d(x,y)\in\{1,2,3,4\}$.
\end{lemma}

Below is an immediate consequence of equation~(\ref{trueoptimal}).

\begin{lemma}\label{optimalisinpreservedregion}
$\hat{p}\in S$.
\end{lemma}

The following lemma is a consequence of
equations~(\ref{neighborhoodinsubgraph})--(\ref{numberoffrozenincidentdistances}) and
our
forbidding
repeated queries.

\begin{lemma}\label{monotonicity}
For all $x\in[n]$ and $i\in[q(n)]$,
\begin{eqnarray*}
\alpha_i(x)-\alpha_{i-1}(x)
=\left\{
\begin{array}{ll}
0, &\text{if $x\notin \{x_i,y_i\}$;}\\
1, &\text{otherwise.}
\end{array}
\right.
\end{eqnarray*}
\end{lemma}
\begin{proof}
The case of $x\notin \{x_i,y_i\}$ is immediate from
equations~(\ref{neighborhoodinsubgraph})--(\ref{numberoffrozenincidentdistances}).
Suppose that $x\in \{x_i,y_i\}$.
By symmetry,
we may assume
$x=x_i$.
So by equation~(\ref{neighborhoodinsubgraph}),
\begin{eqnarray}
N_i(x)=N_{i-1}(x)\cup\left\{y_i\right\}.\label{newneighbor}
\end{eqnarray}
As
$(x,y_i)=(x_i,y_i)$ is the $i$th query
and we
forbid
repeated
queries,
\begin{eqnarray}
y_i\notin N_{i-1}(x)\label{reallynewneighbor}
\end{eqnarray}
by equation~(\ref{neighborhoodinsubgraph}).\footnote{In detail,
if $y_i\in N_{i-1}(x)$, then
$(x_j,y_j)\in\{(x,y_i),(y_i,x)\}$
for some $j\in[i-1]$ by
equation~(\ref{neighborhoodinsubgraph}); hence the $i$th query $(x_i,y_i)
=(x,y_i)$
repeats the $j$th query, a contradiction.}
Equations~(\ref{numberoffrozenincidentdistances})~and~(\ref{newneighbor})--(\ref{reallynewneighbor})
complete the
proof.
\end{proof}

In short, Lemma~\ref{monotonicity} says that adding the edge $(x_i,y_i)$
to an undirected graph without that edge increases the degree of $x$
by $1$ if and only if $x\in\{x_i,y_i\}$.



\begin{lemma}\label{monotonicitysame}
For all $x\in[n]$ and $i\in[q(n)+1]$,
if $\alpha_{i-1}(x)>\delta n$, then
$x\in B$.
\end{lemma}
\begin{proof}
By Lemma~\ref{monotonicity}, $\alpha_{q(n)}(x)\ge \alpha_{i-1}(x)$.
Invoking equation~(\ref{badpoints}) then completes the proof.
\end{proof}

\begin{lemma}\label{sumofdegrees}
$$\sum_{x\in[n]}\, \alpha_{q(n)}(x)= 2\, q(n).$$
\end{lemma}
\begin{proof}
Recall that the left-hand side
is the sum of degrees in the undirected graph with vertex set $[n]$
and edge set $\{(x_i,y_i)\mid i\in[q(n)]\}$.
As we
forbid
repeated queries, $\left|\,\{(x_i,y_i)\mid
i\in[q(n)]\}\,\right|=q(n)$
Finally, it is a basic fact in graph
theory that
the sum of degrees in an undirected graph equals
twice the number of edges.
\end{proof}

\begin{lemma}[{Implicit in~\cite[Lemma~13]{Cha12}}]\label{fewbadpoints}
$|B|=o(n)$.
\end{lemma}
\begin{proof}
We have
$$
|B|\,\delta n
\stackrel{\text{equation~(\ref{badpoints})}}{\le} \sum_{x\in B} \alpha_{q(n)}(x)
\le \sum_{x\in [n]} \alpha_{q(n)}(x)
\stackrel{\text{Lemma~\ref{sumofdegrees}}}{=}
2\,q(n).
$$
This gives $|B|=o(n)$ as $\delta\in(0,0.1)$ is a constant and $q(n)=o(n^2)$.
\end{proof}

\begin{lemma}\label{sparselyaskedpoint}
For all sufficiently large $n$ and all $i\in[q(n)+1]$,
\begin{eqnarray}
\alpha_{i-1}\left(\hat{p}\right)&\le& \delta n.
\label{sparselyaskedpointequation}
\end{eqnarray}
\end{lemma}
\begin{proof}
By Lemma~\ref{fewbadpoints}, $|S|=\lceil\delta n\rceil$ and $\delta\in(0,0.1)$
being a constant,
$S\setminus B\neq\emptyset$ for all sufficiently large $n$.
By equation~(\ref{badpoints}), $S\setminus B\neq\emptyset$
$\alpha_{q(n)}(x)\le \delta n$ for some $x\in S$,
which together with equation~(\ref{trueoptimal})
gives $\alpha_{q(n)}(\hat{p})\le\delta n$.
Finally, Lemma~\ref{monotonicity} and $\alpha_{q(n)}(\hat{p})\le\delta n$
imply inequality~(\ref{sparselyaskedpointequation}) for all $i\in[q(n)+1]$.
\comment{
Furthermore,
$$
|S|\cdot\alpha_{q(n)}\left(\hat{p}\right)
\stackrel{\text{equation~(\ref{trueoptimal})}}{\le}
\sum_{x\in S}\, \alpha_{q(n)}(x)
\le
\sum_{x\in[n]}\, \alpha_{q(n)}(x)
\stackrel{\text{Lemma~\ref{sumofdegrees}}}{=}
2\cdot q(n).
$$
So
$\alpha_{q(n)}(\hat{p})\le 2\cdot q(n)/|S|=o(n)$,
which together with
Lemma~\ref{monotonicity}
shows
$\alpha_{i-1}(\hat{p})
\le\alpha_{q(n)}(\hat{p})
\le\delta n$
for all sufficiently large $n$
and all $i\in[q(n)+1]$.
}
\end{proof}

Henceforth, assume $n$ to be sufficiently large to satisfy
inequality~(\ref{sparselyaskedpointequation}) for all $i\in[q(n)+1]$.

\begin{lemma}\label{onlysourceofdistance1}
For all $x$, $y\in[n]$, if $d(x,y)=1$, then
one of the following conditions is true:
\begin{itemize}
\item $x=\hat{p}$ and $y\notin S\cup B$;
\item $y=\hat{p}$ and $x\notin S\cup B$.
\end{itemize}
\end{lemma}
\begin{proof}
Inspect equation~(\ref{completingthemetric}), which
is the only equation that may set distances to $1$.
\end{proof}

\comment{
Below is a consequence of Lemma~\ref{onlysourceofdistance1}.

\begin{lemma}\label{sourceofshortdistances}
For all distinct
$x$, $y$, $z\in[n]$, if $d(x,y)=d(x,z)=1$, then $x=\hat{p}$ and $y$,
$z\notin S\cup B$.
\end{lemma}
}


\begin{lemma}\label{ordinarydistancesare2}
For all distinct $x$, $y\in[n]\setminus (S\cup B)$,
$d\left(x,y\right)=2$.
\end{lemma}
\begin{proof}
By Lemma~\ref{monotonicitysame},
$\max\{\alpha_{i-1}(x_i),\alpha_{i-1}(y_i)\}>\delta n$
means $\{x_i,y_i\}\cap B\neq\emptyset$, where $i\in[q(n)+1]$.
So
only the second-to-last case
in equation~(\ref{distanceassignment}), which sets $d(x_i,y_i)=2$,
may be consistent
with $x_i$, $y_i\notin S\cup B$.

By
Lemma~\ref{optimalisinpreservedregion},
$\hat{p}\in S$.
So
only the last case
in equation~(\ref{completingthemetric}),
which sets $d(x,y)=2$,
may be consistent
with $x$, $y\notin S\cup B$.
\end{proof}


\comment{
\begin{lemma}\label{askeddistancesincidentonoptimalpoint}
For all $i\in[q(n)]$,
$\alpha_{i-1}(\hat{p})\le\delta n$.
\end{lemma}
\begin{proof}
By Lemma~\ref{sparselyaskedpoint},
$\hat{p}\in S\setminus B$.
$\hat{p}\notin B$, which
together with equation~(\ref{badpoints}) and
Lemma~\ref{monotonicity}
completes the proof.
\end{proof}
}

\begin{lemma}\label{optimalpointdistances1or3}
For all $x\in[n]\setminus\{\hat{p}\}$,
$d(\hat{p},x)\in\{1,3\}$.
\end{lemma}
\begin{proof}
By
Lemma~\ref{optimalisinpreservedregion} and
inequality~(\ref{sparselyaskedpointequation}),
only
the first three
cases in
equation~(\ref{distanceassignment}),
which set $d(x_i,y_i)=3$,
may be consistent with
$x_i=\hat{p}$ or $y_i=\hat{p}$.

Again by Lemma~\ref{optimalisinpreservedregion},
only the first three cases in equation~(\ref{completingthemetric}),
which set $d(x,y)\in\{1,3\}$,
may be consistent with $x=\hat{p}$ or $y=\hat{p}$.
\end{proof}

\begin{lemma}\label{illegaldistances1}
There do not exist distinct $x$, $y$, $z\in[n]$
with $d(x,y)=1$ and $\{d(x,z), d(y,z)\}=\{2, 4\}$.
\end{lemma}
\begin{proof}
By Lemma~\ref{onlysourceofdistance1},
$d(x,y)=1$ implies $\hat{p}\in\{x,y\}$.
By symmetry, assume $x=\hat{p}$.
Then $d(x,z)\in\{1,3\}$ by
Lemma~\ref{optimalpointdistances1or3}.
\end{proof}

\begin{lemma}\label{illegaldistances2}
There do not exist distinct $x$, $y$, $z\in[n]$
with $d(x,y)=d(x,z)=1$ and $d(y,z)\in\{3, 4\}$.
\end{lemma}
\begin{proof}
By Lemma~\ref{onlysourceofdistance1},
$d(x,y)=d(x,z)=1$ implies $x=\hat{p}$ and $y$, $z\notin S\cup B$.
Then $d(y,z)=2$ by
Lemma~\ref{ordinarydistancesare2}.
\end{proof}


Lemmas~\ref{illegaldistances1}--\ref{illegaldistances2}
forbid all possible violations of the triangle inequality,
yielding the following lemma.

\begin{lemma}\label{itismetric}
$([n],d)$ is a metric space.
\end{lemma}
\begin{proof}
Lemmas~\ref{distancesarezeroto4}~and~\ref{illegaldistances1}--\ref{illegaldistances2}
establish the triangle inequality for $d$.
Furthermore, $d$ is symmetric because (1)~freezing $d(x,y)$ automatically
freezes $d(y,x)$ to the same value,
(2)~forbidding
repeated queries
prevents equation~(\ref{distanceassignment})
from assigning inconsistent values to one distance and
(3)~equation~(\ref{completingthemetric}) is symmetric.
All the other axioms for metrics
are easy to verify.
\end{proof}

Recall that $p$ denotes the output of $A^d$.
We proceed to
compare $\sum_{x\in[n]}\,d(p,x)$ with $\sum_{x\in[n]}\,d(\hat{p},x)$.

\begin{lemma}\label{identifyingjumps}
There exist
$k(1)$, $k(2)$, $\ldots$, $k(n-1)\in [q(n)]$
and distinct $z_{k(1)}$, $z_{k(2)}$, $\ldots$, $z_{k(n-1)}\in[n]$
such that
\begin{eqnarray}
\alpha_{k(t)-1}(p)&=&t-1,\label{beforejumping}\\\
\alpha_{k(t)}(p)&=&t,\label{afterjumping}\\
\left(p,z_{k(t)}\right)
&\in&\left\{\left(x_{k(t)}, y_{k(t)}\right), \left(y_{k(t)}, x_{k(t)}\right)
\right\}
\label{algorithmoutputparticipate}
\end{eqnarray}
for all $t\in[n-1]$.
\end{lemma}
\begin{proof}
By
Lemma~\ref{monotonicity}, equation~(\ref{outputallasked}) and the easy fact
that $\alpha_0(p)=0$,
there exist distinct $k(1)$, $k(2)$, $\ldots$, $k(n-1)\in [q(n)]$ satisfying
equations~(\ref{beforejumping})--(\ref{afterjumping})
for all $t\in[n-1]$.\footnote{Observe that
$\alpha_i(p)$ must go through all of $0$, $1$, $\ldots$, $n-1$ as $i$
increases from $0$ to $q(n)$.}
Lemma~\ref{monotonicity} and
equations~(\ref{beforejumping})--(\ref{afterjumping}) imply
$p\in\{x_{k(t)}, y_{k(t)}\}$, establishing the existence of
$z_{k(t)}$ satisfying
equation~(\ref{algorithmoutputparticipate}).
If $z_{k(1)}$, $z_{k(2)}$, $\ldots$, $z_{k(n-1)}$ are not distinct, then
there are repeated queries by
equation~(\ref{algorithmoutputparticipate}), a contradiction.
\end{proof}

From now on, let
$k(1)$, $k(2)$, $\ldots$, $k(n-1)\in [q(n)]$
and distinct $z_{k(1)}$, $z_{k(2)}$, $\ldots$, $z_{k(n-1)}\in[n]$
satisfy
equations~(\ref{beforejumping})--(\ref{algorithmoutputparticipate})
for all $t\in[n-1]$.

\begin{lemma}\label{algorithmoutputtypicaldistances}
For each $t\in[n-1]$, if $t\ge \lceil\delta n\rceil+2$ and $z_{k(t)}\notin S$,
then
$d(p,z_{k(t)})=4$.
\end{lemma}
\begin{proof}
Assume in equation~(\ref{algorithmoutputparticipate}) that
$p=x_{k(t)}$ and $z_{k(t)}=y_{k(t)}$; the other case will be symmetric.
By equation~(\ref{beforejumping}),
\begin{eqnarray}
\alpha_{k(t)-1}\left(x_{k(t)}\right)=t-1> \delta n.
\label{algorithmoutputverymuchasked}
\end{eqnarray}
\begin{enumerate}[{Case }1:]
\item $x_{k(t)}\in S$.
By equation~(\ref{distanceassignment}),
$x_{k(t)}\in S$
and $y_{k(t)}=z_{k(t)}\notin S$,
\begin{eqnarray}
d\left(x_{k(t)},y_{k(t)}\right)
=
\left\{
\begin{array}{ll}
3, &\text{if $\alpha_{k(t)-1}(x_{k(t)})\le \delta n$;}\\
4, &\text{if $\alpha_{k(t)-1}(x_{k(t)})> \delta n$.}
\end{array}
\right.
\label{distancesfromalgorithmoutput1}
\end{eqnarray}
\item $x_{k(t)}\notin S$.
By equation~(\ref{distanceassignment}),
$x_{k(t)}\notin S$
and $y_{k(t)}=z_{k(t)}\notin S$,
\begin{eqnarray}
d\left(x_{k(t)},y_{k(t)}\right)
=
\left\{
\begin{array}{ll}
2, &\text{if $\max\{\alpha_{k(t)-1}(x_{k(t)}),\alpha_{k(t)-1}(y_{k(t)})\}\le \delta n$;}\\
4, &\text{if $\max\{\alpha_{k(t)-1}(x_{k(t)}),\alpha_{k(t)-1}(y_{k(t)})\}> \delta n$.}
\end{array}
\right.
\label{distancesfromalgorithmoutput2}
\end{eqnarray}
\end{enumerate}
Equation~(\ref{algorithmoutputverymuchasked})
together with any one of
equations~(\ref{distancesfromalgorithmoutput1})--(\ref{distancesfromalgorithmoutput2})
implies $d(x_{k(t)},y_{k(t)})=4$.
Hence $d(p,z_{k(t)})=d(x_{k(t)},y_{k(t)})=4$.
\end{proof}

We are now able to analyze the quality of $p$ as a solution to
{\sc metric $1$-median}.

\begin{lemma}\label{analyzingalgorithmoutputassuboptimal}
$$\sum_{x\in[n]}\,d\left(p,x\right)\ge
4\left(n-2\left\lceil\delta n\right\rceil-2\right).$$
\end{lemma}
\begin{proof}
By the distinctness of
$z_{k(1)}$, $z_{k(2)}$, $\ldots$, $z_{k(n-1)}$ in
Lemma~\ref{identifyingjumps},
\begin{eqnarray}
\sum_{x\in[n]}\,d\left(p,x\right)
\ge\sum_{t\in[n-1]}\,d\left(p,z_{k(t)}\right).
\label{initialinequalityinanalyzingalgorithmoutput}
\end{eqnarray}
Write
$A=\{t\in[n-1]\mid z_{k(t)}\in S\}$.
As $z_{k(1)}$, $z_{k(2)}$, $\ldots$, $z_{k(n-1)}$ are distinct,
\begin{eqnarray}
|A|\le |S|.
\end{eqnarray}
Furthermore,
\begin{eqnarray}
&&\sum_{t\in[n-1]}\,d\left(p,z_{k(t)}\right)\nonumber\\
&\ge&\sum_{t\in[n-1],\,t\ge \lceil\delta n\rceil+2,\,t\notin A}\,
d\left(p,z_{k(t)}\right)\nonumber\\
&\stackrel{\text{Lemma~\ref{algorithmoutputtypicaldistances}}}{=}&
\sum_{t\in[n-1],\,t\ge \lceil\delta n\rceil+2,\,t\notin A}\,
4\nonumber\\
&\ge& 4\left(n-\left\lceil\delta n\right\rceil-2-|A|\right).
\label{lastinequalityinanalyzingalgorithmoutput}
\end{eqnarray}
Equations~(\ref{initialinequalityinanalyzingalgorithmoutput})--(\ref{lastinequalityinanalyzingalgorithmoutput})
and $|S|=\lceil\delta n\rceil$ complete the proof.
\end{proof}

We now analyze the quality of $\hat{p}$ as a solution to
{\sc metric $1$-median}.
The following lemma is immediate from
equation~(\ref{completingthemetric}).

\begin{lemma}\label{optimalpointhasmanydistancesbeing1}
For all $y\in [n]\setminus(S\cup B)$,
if
$y\neq \hat{p}$ and
$(\hat{p},y)$, $(y,\hat{p})\notin \{(x_j,y_j)\mid j\in[q(n)]\}$,
then $d(\hat{p},y)=1$.
\end{lemma}

\begin{lemma}\label{analyzingoptimalpoint}
$$\sum_{y\in[n]}\,d\left(\hat{p},y\right)
\le
n+3\cdot\left(\left\lceil\delta n\right\rceil+o(n)+\delta n\right).
$$
\end{lemma}
\begin{proof}
By equation~(\ref{neighborhoodinsubgraph}),
$$N_{q(n)}\left(\hat{p}\right)
=\left\{
y\in[n]\mid
\left\{
\left(\hat{p},y\right), \left(y,\hat{p}\right)
\right\}
\cap \left\{\left(x_j,y_j\right)\mid
j\in\left[q(n)\right]\right\}
\neq\emptyset
\right\}.$$
This and
Lemma~\ref{optimalpointhasmanydistancesbeing1} imply
$d(\hat{p},y)=1$ for all $y\in [n]\setminus(S\cup B)$ with
$y\neq \hat{p}$ and
$y\notin N_{q(n)}(\hat{p})$.
Therefore,
\begin{eqnarray}
\sum_{y\in [n]\setminus(S\cup B\cup
N_{q(n)}(\hat{p}))}\,d\left(\hat{p},y\right)
\le
n-\left|\,S\cup B\cup
N_{q(n)}\left(\hat{p}\right)\,\right|.
\label{distance1parts}
\end{eqnarray}
Clearly,
\begin{eqnarray}
\sum_{y\in S\cup B\cup N_{q(n)}(\hat{p})}\,d\left(\hat{p},y\right)
\stackrel{\text{Lemma~\ref{distancesarezeroto4}}}{\le}
\sum_{y\in S\cup B\cup N_{q(n)}(\hat{p})}\,4
= 4\cdot\left|\,S\cup B\cup
N_{q(n)}\left(\hat{p}\right)\,\right|
\label{largedistancesparts}
\end{eqnarray}
Furthermore,
\begin{eqnarray}
\left|\,N_{q(n)}\left(\hat{p}\right)\,\right|
\stackrel{\text{equation~(\ref{numberoffrozenincidentdistances})}}{=}
\alpha_{q(n)}\left(\hat{p}\right)
\stackrel{\text{inequality~(\ref{sparselyaskedpointequation})}}{\le}
\delta n.\nonumber
\end{eqnarray}
This and Lemma~\ref{fewbadpoints}
imply
\begin{eqnarray}
\left|\,S\cup B\cup N_{q(n)}\left(\hat{p}\right)\,\right|
\le \left\lceil\delta n\right\rceil+o(n)+\delta n
\label{fewneighborsforoptimalpoint}
\end{eqnarray}
as $|S|=\lceil\delta n\rceil$.
To complete the proof, sum
up
inequalities~(\ref{distance1parts})--(\ref{largedistancesparts})
and then use
inequality~(\ref{fewneighborsforoptimalpoint})
in the trivial way.
\end{proof}

Combining
Lemmas~\ref{itismetric},~\ref{analyzingalgorithmoutputassuboptimal}~and~\ref{analyzingoptimalpoint}
yields our main theorem, stated below.

\begin{theorem}\label{maintheorem}
{\sc Metric $1$-median} has no deterministic $o(n^2)$-query
$(4-\epsilon)$-approximation algorithm for any constant $\epsilon>0$.
\end{theorem}
\begin{proof}
Lemma~\ref{itismetric} asserts that $([n],d)$ is a metric space.
By
Lemmas~\ref{analyzingalgorithmoutputassuboptimal}~and~\ref{analyzingoptimalpoint},
$$
\sum_{x\in[n]}\,d\left(p,x\right)
\ge 4\left(1-8\delta-o(1)\right)\sum_{x\in[n]}\, d\left(\hat{p},x\right).
$$
This proves the theorem because
the deterministic $o(n^2)$-query algorithm $A$
and
the constant $\delta\in(0,0.1)$
are picked arbitrarily (note that $p$ denotes the output of $A^d$).
\end{proof}

Theorem~\ref{maintheorem}
complements
Theorem~\ref{nonadaptiveupperbound}.

It is possible to
simplify equation~(\ref{completingthemetric})
at the expensive of an additional assumption.
Without loss of generality, we may assume that $\alpha_{q(n)}(x)=n-1$
for all $x\in B$; this increases the query complexity by a multiplicative
factor of $O(1)$ by equation~(\ref{badpoints}).
Therefore, if $x\in B$ or $y\in B$, then
$d(x,y)$ will be frozen by equation~(\ref{distanceassignment}).
So the third to fifth
cases
in equation~(\ref{completingthemetric}),
which satisfies $x\in B$ or $y\in B$,
can be omitted.


\bibliographystyle{plain}
\bibliography{adaptivemedian}

\begin{thebibliography}{10}

\bibitem{AGKMMP04}
V.~Arya, N.~Garg, R.~Khandekar, A.~Meyerson, K.~Munagala, and V.~Pandit.
\newblock Local search heuristics for $k$-median and facility location
  problems.
\newblock {\em SIAM Journal on Computing}, 33(3):544--562, 2004.

\bibitem{Cha12}
C.-L. Chang.
\newblock Some results on approximate $1$-median selection in metric spaces.
\newblock {\em Theoretical Computer Science}, 426:1--12, 2012.

\bibitem{Cha13}
C.-L. Chang.
\newblock Deterministic sublinear-time approximations for metric $1$-median
  selection.
\newblock {\em Information Processing Letters}, 113(8):288--292, 2013.

\bibitem{Che09}
K.~Chen.
\newblock On coresets for $k$-median and $k$-means clustering in metric and
  \uppercase{E}uclidean spaces and their applications.
\newblock {\em SIAM Journal on Computing}, 39(3):923--947, 2009.

\bibitem{EW04}
D.~Eppstein and J.~Wang.
\newblock Fast approximation of centrality.
\newblock {\em Journal of Graph Algorithms and Applications}, 8(1):39--45,
  2004.

\bibitem{GR08}
O.~Goldreich and D.~Ron.
\newblock Approximating average parameters of graphs.
\newblock {\em Random Structures \& Algorithms}, 32(4):473--493, 2008.

\bibitem{GMMMO03}
S.~Guha, A.~Meyerson, N.~Mishra, R.~Motwani, and L.~O'Callaghan.
\newblock Clustering data streams: \uppercase{T}heory and practice.
\newblock {\em IEEE Transactions on Knowledge and Data Engineering},
  15(3):515--528, 2003.

\bibitem{Ind99}
P.~Indyk.
\newblock Sublinear time algorithms for metric space problems.
\newblock In {\em Proceedings of the 31st Annual ACM Symposium on Theory of
  Computing}, pages 428--434, 1999.

\bibitem{Ind00}
P.~Indyk.
\newblock {\em High-dimensional computational geometry}.
\newblock PhD thesis, Stanford University, 2000.

\bibitem{JKS12}
R.~Jaiswal, A.~Kumar, and S.~Sen.
\newblock A simple $\uppercase{D}^2$-sampling based \uppercase{PTAS} for
  $k$-means and other clustering problems.
\newblock In {\em Proceedings of the 18th Annual International Conference on
  Computing and Combinatorics}, pages 13--24, 2012.

\bibitem{KSS10}
A.~Kumar, Y.~Sabharwal, and S.~Sen.
\newblock Linear-time approximation schemes for clustering problems in any
  dimensions.
\newblock {\em Journal of the ACM}, 57(2):5, 2010.

\bibitem{MP04}
R.~R. Mettu and C.~G. Plaxton.
\newblock Optimal time bounds for approximate clustering.
\newblock {\em Machine Learning}, 56(1--3):35--60, 2004.

\bibitem{WF94}
S.~Wasserman and K.~Faust.
\newblock {\em Social Network Analysis: Methods and Applications}.
\newblock Cambridge University Press, 1994.

\bibitem{Wu14}
B.-Y. Wu.
\newblock On approximating metric $1$-median in sublinear time.
\newblock {\em Information Processing Letters}, 114(4):163--166, 2014.

\end{thebibliography}

\noindent

\end{document}